\documentclass[%
 reprint,
 amsmath,amssymb,
 aps,
]{revtex4-2}

\usepackage{graphicx}
\usepackage{dcolumn}
\usepackage{bm}
\usepackage{quiver}
\usepackage{indentfirst}
\usepackage{quiver}
\usepackage{amsfonts}
\usepackage{amsmath}
\usepackage{mathrsfs}
\usepackage{amsthm}
\usetikzlibrary{arrows}
\theoremstyle{plain}
\newtheorem{theorem}{Theorem}[section]

\theoremstyle{definition}
\newtheorem{definition}{Definition}[section]

\DeclareMathOperator{\Tr}{Tr}

\begin{document}

\preprint{APS/123-QED}

\title{A categorical view of Bell's inequalities in quantum field theory}

\author{Rafael Grossi}
\email{rgrossi@usp.br}
\author{João C. A. Barata}%
 \email{jbarata@if.usp.br}
\affiliation{%
 Departmento de Física-Matemática,\\
Instituto de Física - Universidade de São Paulo,\\
R. do Matão 1371,
Cidade Universitária, São Paulo, Brazil
}%

\date{\today}

\begin{abstract}
We propose a generalization of the description of Bell's inequalities in algebraic quantum field theory (AQFT) to the context of locally covariant quantum field theory (LCQFT). We use the functorial formulation of the state space as proposed in the seminal work \cite{BFV03} to show that for the suitable subcategory which consists of all possible admissible quadruples, yields a category of states over such observables which satisfy the Clauser-Horne-Shimony-Holt (CHSH) inequality. Such inequality is trivially preserved under algebraic morphisms and functoriality of the state space asserts us that it is preserved under isometric embeddings between the globally hyperbolic spacetimes where the algebras are defined and which constitute the locally covariant QFT functor.
\end{abstract}

\maketitle


\section{\label{sec:intro}Introduction}

Since its introduction in a pair of papers by John Bell \cite{Bell:1964,Bell:1966}, the set of results that became known as \emph{Bell's inequalities} have been important to our modern understanding of quantum physics and have found direct applications in quantum information and quantum computing. For a recent review of such applications, see reference \cite{review_bell}. In recent collider experiments, violations of Bell's inequalities have been achieved, which indicates the importance of this topic to phenomenology \cite{barr2024quantum}. For a historical overview and a philosophical perspective, see reference \cite{sep-bell-theorem}.

These inequalities allowed physicists to first probe experimentally the validity of a physical description in terms of hidden variables in an atomic or subatomic scale. They consist of statistical correlations of measurements made on two parts of one system and their violation reveals the fundamental difference between classical and quantum regimes, with the presence of entanglement in the latter.

Bell's inequalities have also played a role in understanding entanglement in quantum field theory \cite{RevModPhys_entanglement}. In the context of algebraic quantum field theory \cite{haag}, the work by Summers and Werner \cite{summers&werner_physlett,summers&werner1,summers&werner2,Summers:1987ze,summers&werner_tangent} has shown that the vacuum state violates \emph{maximally} Bell's inequalities, which translates as it being a maximally entangled state. This reveals maximal correlations in the vacuum fluctuations which are known to be a property of quantum field theories as a consequence of the Reeh-Schlieder property \cite{Reeh:1961ujh,witten_entanglement}. Recent numerical work has shown that the inequalities hold in a variety of settings \cite{Dudal:2023mij,DeFabritiis:2023llu,DeFabritiis:2023tkh,DeFabritiis:2024jfy}.

Alongside this development, progress has been made in understanding quantum field theory, in particular, through the mathematical formulation of QFT in terms of a functor between suitable categories. This functorial axiomatization process began in the context of topological quantum field theories (TQFT) in the work by Atiyah \cite{atiyah_tqft} and was introduced in the context of algebraic quantum field theory (AQFT) based on the Haag-Kastler-Araki axioms by Brunetti, Fredenhagen, and Verch \cite{BFV03}. In this paper, we will consider the latter and refer to it by \emph{locally covariant quantum field theory} (LCQFT), as it has become known in the literature.

The LCQFT framework is described by a covariant functor between the category of globally hyperbolic spacetimes with isometric embeddings as morphisms and the category of algebras of observables with structure preserving homomorphisms as morphisms. Such construction is presented as the commuting diagram
\[\begin{tikzcd}
	{(M,\textbf{g})} && {(M',\textbf{g}')} \\
	\\
	{\mathfrak{A}(M,\textbf{g})} && {\mathfrak{A}(M',\textbf{g}')}
	\arrow["\psi", from=1-1, to=1-3]
	\arrow["{\mathfrak{A}}", from=1-3, to=3-3]
	\arrow["{\alpha_\psi}"', from=3-1, to=3-3]
	\arrow["{\mathfrak{A}}"', from=1-1, to=3-1]
\end{tikzcd}\]
where $\psi$ denotes the isometric embedding between manifolds $(M,\mathbf{g})$ and $(M',\mathbf{g'})$ and $\alpha_\psi$ the corresponding morphism (i.e., the algebraic $*$-monomorphism $\gamma: \mathfrak{A}\rightarrow\mathfrak{A}'$, with $\mathfrak{A} \doteq\mathfrak{A}(M,\mathbf{g})$ and $\mathfrak{A}'\doteq \mathfrak{A}(M',\mathbf{g'})$).

In general, it suffices to consider abstract $C^*$-algebras as being the algebra of observables. A more concrete description in terms of Hilbert spaces and bounded operators acting on them can then be obtained via the GNS construction theorem. The main advantage of using such abstract mathematical language to treat QFT is the model-independent character of such description, allowing one to make very general statements about different (free) theories. Moreover, this language allows to interpret quantum fields as natural transformations between suitable functors and to define properly a QFT in general globally hyperbolic spacetimes (for further detail, see the original work in \cite{BFV03}).

The goal of this work is thus to extend the formulation and results of Summers and Werner to this more general framework which can be seen as an extension of the results to include curved spacetimes. In section \ref{sec:space_of_states}, we review the formulation of the space of states in the context of LCQFT, which is also done in terms of a functor between suitable categories. In section \ref{sec:bell_aqft}, we review some of the main points and definitions of the work of Summers and Werner and state without proof the theorem which states that the violation of Bell's inequalities must occur in theories without hidden variables. The categorical formulation of such violations is finally introduced in section \ref{sec:categorical_bell} by the restriction of the algebra of observables to the $C^*$-algebra elements which satisfy the conditions of the theorem and conclude that any QFT which can be formulated in terms of a functor as described in \cite{BFV03} must have violations of Bell's inequalities, extending thus the result of Summers and Werner to more general spacetimes. Finally, in section \ref{sec:speculations}, we speculate on some possible consequences of this generalization and propose some future directions for this work. 

\section{The space of states in LCQFT}\label{sec:space_of_states}
In this section, we provide the formulation of the space of states as a functor from the category of globally hyperbolic spacetimes to the category of states over the algebra of observables of our theory. This was first proposed in \cite{BFV03} and we follow their notation. 

First, we recall the definition of a state in the context of $C^*$-algebras.

\begin{definition}
    Let $\mathfrak{A}$ denote a unital $C^*$-algebra. A \emph{state over} $\mathfrak{A}$ is a linear functional $\omega: \mathfrak{A}\rightarrow \mathbb{C}$ having the property of being positive, i.e., $\omega(A^{*}A) \geq 0$, $\forall A\in\mathfrak{A}$ and normalized, i.e., $\omega(\mathbf{1}_{\mathfrak{A}}) = 1$, $\mathbf{1}_{\mathfrak{A}}$ being the identity in the $C^*$-algebra $\mathfrak{A}$.
    
    The collection of all states over an algebra $\mathfrak{A}$ is called a \emph{state space} $\mathbf{S}$. This construction can be done in terms of a functor.
\end{definition}

\subsection{The category of state spaces}

We define $\mathfrak{Sts}$ as the category whose objects $\mathbf{S}$ are state spaces of a unital $C^*$-algebra $\mathfrak{A}$ which are closed under taking finite convex combinations and operations $\omega (.)\mapsto \omega_A(.) = \omega(A^* . A)/\omega(A^* A)$, $A\in\mathfrak{A}$, $A\neq 0$ and $\omega(A^*A)\neq 0$. In this language, the algebra is given in terms of the functor $\mathscr{A}$, defined in the introduction.

The morphisms between objects $\mathbf{S'}$ and $\mathbf{S}$ are maps $\gamma^*: \mathbf{S'}\rightarrow \mathbf{S}$ which arise as dual maps from the corresponding $C^*$-monomorphisms from the underlying algebras, $\gamma: \mathfrak{A}\rightarrow\mathfrak{A'}$. This map is defined as
\begin{equation}
    \gamma^*\omega'(A) = \omega'(\gamma(A)), \hspace{0.25cm}\omega'\in\mathbf{S'},\hspace{0.25cm} A\in\mathfrak{A}.
    \label{eq:dual_definition}
\end{equation}
Notice the inverted order of the domain and codomain of $\gamma$ and $\gamma^*$.

\subsection{The functor}
As presented in the introduction, one can define a QFT as a functor between a category of spacetimes and a category of algebras of observables (usually taken as unital $C^*$-algebras).

In a similar fashion, one can define the functor which assigns to each spacetime the corresponding space of states.

\begin{definition}
    Let $\mathscr{A}$ be a LCQFT. A \emph{state space for} $\mathscr{A}$ is a contravariant functor $\mathscr{S}$ between $\mathfrak{GlobHyp}$ and $\mathfrak{Sts}$:
    \[\begin{tikzcd}
	{(M,\mathbf{g})} && {(M',\mathbf{g'})} \\
	\\
	{\mathscr{S}(M,\mathbf{g})} && {\mathscr{S}(M',\mathbf{g'})}
	\arrow["\psi", from=1-1, to=1-3]
	\arrow["{\mathscr{S}}"', from=1-1, to=3-1]
	\arrow["{\mathscr{S}}", from=1-3, to=3-3]
	\arrow["{\alpha^*_\psi}", from=3-3, to=3-1]
    \end{tikzcd}\]
    where $\mathscr{S}(M,\mathbf{g})\doteq \mathbf{S}(M,\mathbf{g})$ is the set of states on $\mathfrak{A}(M,\mathbf{g})$ and $\alpha^*_{\psi}$ is the dual map of $\alpha_\psi$. The contravariance property is given by
    \begin{equation}
        \alpha^*_{\Tilde{\psi}\circ\psi} = \alpha^*_\psi\circ\alpha^*_{\Tilde{\psi}}.
    \end{equation}
\end{definition}

We note that the inversion of the domain and codomain of each $C^*$-monomorphism, as observed in the previous section, is reflected in the contravariance nature of the state space functor.

\section{Bell's inequalities in AQFT}\label{sec:bell_aqft}
In this section, we present some of the key concepts in formulating Bell's inequalities (in particular, the so called CHSH inequality \cite{Clauser:1969ny}) in the context of AQFT. The notation and conventions here follow references \cite{summers&werner1, summers&werner2, Summers:1987ze, LANDAU198754, Baez:1987kzw}, where violations of the inequalities in the context of AQFT were first investigated.

\subsection{General formulation}
Bell's inequalities can be described as a constraint in the statistical correlations between measurements performed in two different causally disconnected subregions. In the particular case of quantum field theory, one may consider causally disjoint regions of spacetime where different local algebras can be defined on via the LCQFT functor.

In a more general setting, we consider a \emph{order-unit space} $(\mathfrak{A}, \geq, 1_{\mathfrak{A}})$, that is, a vector space with an ordering and a unit. This ordering will be given by a convex cone:
\begin{equation}
    \mathfrak{A}_+ \equiv \{A\in\mathfrak{A}:A\geq 0\}.
\end{equation}
In the case of unital $C^*$-algebras, we restrict to the set of positive elements, which are the self-adjoint elements of the algebra whose spectrum is a subset of the real positive line. A standard theorem (\cite{bratteli1}, theorem 2.2.10) tells us that an element $A$ is positive if, and only if, there is an element $B$ such that $A = B^*B$.

In this setting, a measurement is formalized as a family $\{A_\alpha\}_{\alpha\in I}$ where the index set $I$ represents the possible outcomes of the measurement and $A_\alpha \in \mathfrak{A}_+$. Thus, we have the condition that
\begin{equation}
    \sum_{\alpha\in I}A_\alpha = 1.
    \label{eq:locality_condition}
\end{equation}
This allows us to interpret the state evaluated at $A_\alpha$, that is, $\omega(A_\alpha)$, as the probability of obtaining $\alpha$ in an experiment, where the preparing device is given by the state $\omega$ and the measuring device is given by the family $\{A_\alpha\}_{\alpha\in I}$.

\subsection{Correlations}

With the basic notation in order, we can now proceed to define the correlations we will use to describe the inequalities in the algebraic setting.

\begin{definition}
     A \emph{correlation duality} consists of two order-unit spaces $\mathfrak{A}$ and $\mathfrak{B}$ together with a bilinear functional $\hat{p}: \mathfrak{A}\times\mathfrak{B}\rightarrow\mathbb{R}$ such that $A\in\mathfrak{A}$ and $B\in\mathfrak{B}$ with $A,B\geq 0$ imply $\hat{p}(A,B)\geq 0$ and $\hat{p}(1_{\mathfrak{A}}, 1_{\mathfrak{B}}) = 1$.
\end{definition}

In the case of $C^*$-algebras, which we specialize to, the functional $\hat{p}$ is given by a state on the larger algebra $\mathfrak{C} = \mathfrak{A}\otimes\mathfrak{B}$,
\begin{equation}
    \hat{p}(A,B)\equiv \omega(AB),
\end{equation}
where the product $AB$ is understood to be taken in $\mathfrak{C}$. Henceforth we shall restrict ourselves to the $C^*$-algebraic setting.

Using the intuition gained in the last section, we can interpret the correlation duality as a prescription which gives the probability (the functional $\omega$) of measuring $A_i$ and $B_j$ in the sites $\mathfrak{A}$ and $\mathfrak{B}$, respectively.

This structure naturally encapsulates a notion of \emph{locality}: bilinearity of $\hat{p}$ alongside equation \eqref{eq:locality_condition} implies that for any pair of measuring devices $\{A_i\}_{i\in I}$, $\{A'_j\}_{j\in J}$ in $\mathfrak{A}$ and any $B\in\mathfrak{B}$, we have
\begin{equation*}
    \sum_i \hat{p}(A_i, B) = \sum_j \hat{p}(A'_j, B) = \hat{p}(1_{\mathfrak{A}}, B).
\end{equation*}

This means that the probability of an outcome at site $\mathfrak{B}$ does not depend on the outcome at $\mathfrak{A}$, an assumption usually taken in the derivation of Bell's inequalities.

In the context of $C^*$-algebras, this locality condition is a condition of Einstein causality, which is naturally present in the Haag-Kastler-Araki axioms of AQFT \cite{haag}.

Finally, we consider measurements with two possible outcomes, say $\{-1, +1\}$ which can be measured by two measuring devices $A_+$ and $A_-$ localized at different regions. In this case, we can establish a one-to-one correspondence with generic elements $A\in \mathfrak{A}$ with $-1_{\mathfrak{A}}\leq A\leq 1_{\mathfrak{A}}$ via $A_{\pm} = \frac{1}{2}(1\pm A)$. This motivates giving the following definition:

\begin{definition}
    Given two $C^*$-algebras $\mathfrak{A}$ and $\mathfrak{B}$, an \emph{admissible quadruple} is a set $\{A_1, A_2, B_1, B_2\}$ with $A_i \in \mathfrak{A}$, $B_i\in\mathfrak{B}$, $i\in \{1,2\}$, $\{\mathfrak{A},\mathfrak{B}\} = \{0\}$ and $-1_{\mathfrak{A}}\leq A_i \leq 1_{\mathfrak{A}}$ and similar for $B_i$. 
\end{definition}

An admissible quadruple is said to satisfy Bell's inequality if 
\begin{equation}
    |\omega(A_1 B_1) + \omega(A_1 B_2) + \omega(A_2 B_1) - \omega(A_2 B_2)|\leq 2\sqrt{2},
    \label{eq:bell_ineq}
\end{equation}
which is the algebraic version of Bell's inequality usually derived assuming that the order-unit spaces are classical.

One can use the linearity of the states to write equation \eqref{eq:bell_ineq} compactly as
\begin{equation}
    \omega(\mathcal{C})\leq 2\sqrt{2},
\end{equation}
where $\mathcal{C} = A_1B_1 + A_1 B_2 + A_2B_1 - A_2B_2$, which is referred to in the literature as Bell's operator.

We present an important theorem based upon references \cite{summers&werner1, LANDAU198754}, which tells us that there are functionals such as in \eqref{eq:bell_ineq} which can violate maximally Bell's inequalities when we consider von Neumann algebras.

\begin{theorem}[Summers, Werner \& Landau]
    Let $\mathfrak{M}$ be a von Neumann algebra with two independent subalgebras $\mathfrak{M}_1$ and $\mathfrak{M}_2$. Then, there exists an admissible quadruple $\{A_1, A_2, B_1, B_2\}$ such that $[A_1, A_2]\neq 0$ and $[B_1, B_2]\neq 0$ (hence satisfying the hypotheses for the first Bell inequality) and a pure state $\omega_0$ over $\mathfrak{M}$ such that we have a \emph{maximal violation of Bell's inequality}:
    \begin{equation}
        |\omega_0(\mathcal{C})| = 2\sqrt{2}.
    \end{equation}
\end{theorem}

\begin{proof}
    We begin by showing by construction that we can find the admissible quadruple $\{A_1, A_2, B_1, B_2\}$ satisfying the hypotheses for the fist Bell inequality. Since any element of a von Neumann algebra can be written as a sum of self-adjoint elements, we can always find orthogonal projections which do not commute with one another, which are the spectral projections associated with the self-adjoint elements. If the algebra is non-commutative, these projections must indeed commute with each other.

    Hence, consider $E$ and $F$ as orthogonal projections which do not commute with one another. Define the element $T\in\mathfrak{M}$ by
    \begin{equation}
        T \equiv EF(\mathbf{1} - E).
    \end{equation}
    This element is certainly non-vanishing, otherwise we would have that $EF = EFE \implies FE = EFE \implies EF = FE$, by self-adjointness, which is a contradiction. On the other hand, by an elementary computation, we see that $T^2 = 0$. This means that the range of $T$ is a subset of its kernel.

    Let $T = V|T|$ be the polar decomposition of $T$, with partial isometry $V$, which is contained in the von Neumann algebra. This isometry satisfies $\ker (V) = \ker (T)$ and $\text{ran}(V) = \overline{\text{ran}(T)}$ and from $\text{ran}(T)\subset \ker(T)$, we have that $V^2 = 0$. It is easy to see that the partial isometry satisfies
    \begin{equation}
        VV^*V = V,\hspace{0.2cm} V^*VV^* = V^*,
    \end{equation}
    since $V^*V$ is an orthogonal projection over $(\ker(V)^\perp$. Now, we define
    \begin{equation}
        X\equiv V^*V,\hspace{0.2cm} Y\equiv VV^*,\hspace{0.2cm} Z\equiv X + Y.
    \end{equation}
    The operators $X$ and $Y$ are trivially projection operators and they are orthogonal to one another. From this, it follows that $Z$ is an orthogonal projection (by the orthogonal decomposition theorem). Next, define
    \begin{equation}
        \overline{A}_1 \equiv V + V^*, \hspace{0.2cm} \overline{A}_2 \equiv i(V^* - V).
    \end{equation}
    Again, an elementary computation shows that $\overline{A}_1^2 = \overline{A}_2^2 = Z$ and $[\overline{A}_1, \overline{A}_2] = 2i (Y- Z)$. From the $C^*$-property (recalling that a von Neumann algebra is in particular a $C^*$-algebra), it follows that $||(Y- X)||^2 = ||(Y- X)^2|| = ||X+Y||^2 = ||Z||^2 = 1$, which implies
    \begin{equation}
        ||[\overline{A}_1, \overline{A}_2]|| = 2.
        \label{eq:overline_commutator}
    \end{equation}
    Finally, define
    \begin{equation}
        A_1 = \overline{A}_1 + Z - \mathbf{1},\hspace{0.2cm} A_2 = \overline{A}_2 + Z - \mathbf{1}.
    \end{equation}
    By a simple computation, using $\overline{A}_1^2 = \overline{A}_2^2 = Z$ and equation \eqref{eq:overline_commutator}, it follows that $A_1^2 = A_2^2 = \mathbf{1}$ and that $||[A_1, A_2]|| = 2$. From the first result, it follows that the spectrum of $A_1$ and $A_2$ must be $\{- 1, +1\}$ and from the second it follows that $A_1 \neq \mathbf{1}$ and similar for $A_2$. The proof for $B_1$ and $B_2$ is identical.

    From squaring Bell's operator, it follows that
    \begin{equation}
        ||\mathcal{C}|| = 2\sqrt{1 + \frac{1}{4}||[A_1, A_2]\,[B_1, B_2]||} = 2\sqrt{2},
    \end{equation}
    with $A_1, A_2, B_1, B_2$ defined as the construction above. By a standard state norm correspondence lemma, there is a state $\omega_0$ such that
    \begin{equation}
        |\omega_0(\mathcal{C})| = 2\sqrt{2},
    \end{equation}
    thus completing the proof.
\end{proof}

\subsection{Violation of the inequalities in the vacuum}
In \cite{summers&werner2}, the authors prove a surprising result: in the Wedge regions in Minkowski spacetime, observables taken in the vacuum state $\omega_0$ violate \emph{maximally} Bell's inequalities, meaning that
\begin{equation}
    \beta(\omega_0,\mathfrak{A}(\mathscr{W}),\mathfrak{A}(\mathscr{W}')) = \sqrt{2},
\end{equation}
where $\mathscr{W}$ and $\mathscr{W}'$ are each other's causal complement and are called \emph{complementary wedges}. Concretely:
\begin{equation}
    \mathscr{W} = \{x\in M : x_1 > |x_0|\},\quad \mathscr{W}' = \{x\in M : -x_1 > |x_0|\}
\end{equation}
where $M$ denotes Minkowski spacetime, i.e., $\mathbb{R}^4$ equipped with the Minkowski metric tensor.

In the case of the free boson quantum field of mass $m\geq 0$, $\phi(.)$, the vacuum state can be defined from the bilinear symmetric form $q$ on the test function space $C_0^{\infty}(M)$ which acts as the domain for the field:
\begin{equation}
    \omega_0(W(f)) = \exp\left({-\frac{1}{4}q(f,f)}\right),\quad f\in C_0^{\infty}(M),
\end{equation}
where $W(f) = \exp{i\phi(f)}$ is the unitary Weyl operator, which defines the canonical commutation relations (CCR) of the theory \cite{bratteli2}. In terms of a concrete Hilbert space $\mathfrak{H}$, we assume that there is a cyclic vector $\Omega \in\mathfrak{H}$ such that $\langle\Omega,W(f)\Omega\rangle = \omega_0(W(f))$. Such a vector is guaranteed to exist by the GNS construction theorem \cite{reed&simon}.

We will not go into the details of the original proof for this result and rather take it as a given result for our discussion. We do note, however, that this violation is not exclusive to causal wedge regions: in \cite{Summers:1987ze}, the authors extended the discussion to diamond regions in Minkowski spacetime and include the different classifications of the von Neumann algebras involved \cite{sunder2012invitation}.

\section{Categorical formulation of Bell's Inequalities}\label{sec:categorical_bell}

In this section, we propose a formulation of the preceding results in terms of the functorial space state formulation \cite{BFV03}.

Let $\mathfrak{Bell}$ denote the subcategory of $\mathscr{S}(M,g)$ composed by the states over the algebras generated by admissible quadruples which satisfy Bell's inequalities, as presented in equation \eqref{eq:bell_ineq}. 

Since we constructed the category $\mathscr{S}(M,g)$ as being closed by convex combinations, this subcategory is well defined and it inherits the same morphisms between its objects.

Making use of \eqref{eq:dual_definition}, we obtain trivially
\begin{equation}
    \omega(AB) = (\gamma^*\omega')(AB) = \omega'(\gamma(A)\gamma(B)),
\end{equation}
where we have used the fact that $\gamma$ is a $*$-monomorphism. The first equality is justified by our restriction to the range of $\gamma$ in the algebra $\mathfrak{B}\in\text{ob}(\mathfrak{Bell})$, which defines a surjective map. Thus, we obtain the inequality
\begin{widetext}
\begin{equation}
    |\omega'(\gamma(A_1)\gamma(B_1)) + \omega'(\gamma(A_1)\gamma(B_2)) + \omega'(\gamma(A_2)\gamma(B_1)) - \omega'(\gamma(A_2)\gamma(B_2))|\leq 2\sqrt{2}.
    \label{eq:computation_bell_functor}
\end{equation}
\end{widetext}
In particular, if we take $\mathscr{S}(M,g)$ to be the state space generated over the Weyl algebra over the wedge regions and take the vacuum state, the maximal Bell correlation must be preserved, that is, the supremum is still attained.

This discussion can be specialized to the case of \emph{normal states}. These are the states that can be written as
\begin{equation}
    \omega_{\rho}(AB) = \Tr(\rho\pi(A)\pi(B)),
\end{equation}
where $\pi$ is the GNS representation associated to the state $\omega_\rho$ and $\rho$ is a density matrix (i.e., a trace-class operator on the GNS Hilbert space $\mathcal{H}$ with $\Tr(\rho) = 1$).

In this case, the analysis should happen at the level of the local folia of the GNS representation of such states assuming that such states satisfy a variant of the microlocal spectrum condition, as discussed in \cite{BFV03}. Recall that the \emph{folium} of a representation $\pi$ denoted by $\mathbf{F}(\pi)$ is the set of all normal states which can be written in terms of that representation. We then say that two states $\omega$ and $\Tilde{\omega}$ (or their GNS-representations $\pi$ and $\Tilde{\pi}$) are \emph{locally quasi-equivalent} if
\begin{equation*}
    \mathbf{F}(\pi\circ\alpha_{M,O}) = \mathbf{F}(\Tilde{\pi}\circ\alpha_{M,O}),
\end{equation*}
for all $O\in\mathcal{K}(M,\mathbf{g})$ (the set of all relatively compact subsets of $M$ with causal curves connecting each pair of points) and $\alpha_{M,O} = \alpha_{\imath_{M,O}}$ and $\imath_{M,O}:(O,\mathbf{g}_O)\rightarrow (M,\mathbf{g})$ is the natural embedding.

We also say that a state $\omega$ is \emph{locally normal} to $\Tilde{\omega}$ if
\begin{equation*}
    \omega\circ\alpha_{M,O}\in\mathbf{F}(\Tilde{\pi}\circ\alpha_{M,O}),
\end{equation*}
for all $O\in\mathcal{K}(M,\mathbf{g})$.

It is an easy consequence of the computation in equation \eqref{eq:computation_bell_functor} that since a normal state $\omega_\rho$ violates the CHSH inequality, than so does any locally equivalent or locally normal state $\Tilde{\omega}$. In fact, the simultaneous violation of the inequality by locally equivalent states is a trivial consequence of theorem 2.3 of \cite{Summers:1987ze}, where it is found that any normal state attains the supremum of inequality \eqref{eq:bell_ineq}. For locally normal states, the violation follows from the composition with $\alpha_{M,O}$.

\subsection{Physical Interpretation}
The previous result is illuminating in the following sense. Functoriality of the state space and of the formulation of quantum field theory in globally hyperbolic spacetimes asserts us that the inequalities must be preserved when we move from one state space to the other (or correspondingly, from one globally hyperbolic spacetime to the other). This indicates that \emph{in any quantum field theory that admits a formulation in terms of an algebra of observables, ther must be a set of observables which satisfy Bell's inequalities in the sense of equation} \eqref{eq:bell_ineq}.

Furthermore, if a state violates Bell's inequalities maximally in one theory described by a $C^*$-algebra, then functoriality guarantees that the corresponding state in the target state space must also violate maximally the inequalities. In other words, this indicates that \emph{maximal violations of Bell inequalities are a general feature of quantum field theories}.

These results seem to indicate that the formulation of quantum field theory both as a $C^*$-algebra and as a functor give rise naturally to Bell's inequalities and to violations of such.

\subsection{An example in $d = 1 + 1$ dimensions}

To illustrate the results in this section, we present an example in $1 + 1$ dimensions which can be readily generalized to higher dimensions. We consider the right Rindler wedge region of the two-dimensional Minkowski spacetime $\mathbb{M}_2$, defined by
\begin{equation}
    \mathscr{W} = \{(X^0, X^1)\in \mathbb{M}_2: X^0\leq |X^1|\}.
\end{equation}
In this region and in the corresponding left Rindler wedge, much progress has been made in understanding the states in AQFT (see, for instance, \cite{BisognanoWichmann1,BisognanoWichmann2}). We will be particularly interested in the ``translated Rindler wedge'', denoted as $\mathscr{W} + \epsilon$, obtained by translating $\mathscr{W}$ by a spacelike parameter $\epsilon$ completely contained in $\mathscr{W}$, as depicted in figure \ref{fig:wedge}.

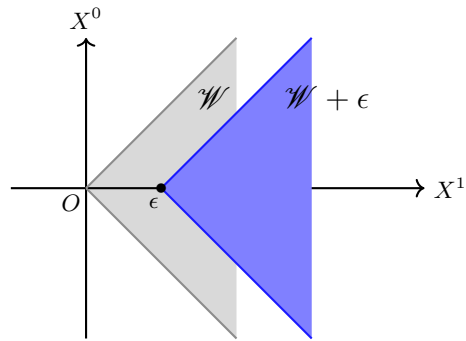
\begin{figure}[!h]
   \begin{tikzpicture}

\fill[fill=gray!30, opacity=0.7] (0,0) -- (2,2) -- (2,-2) -- cycle;  

\draw[thick,->] (-1,0) -- (4.5,0) node[right] {$X^1$};  
\draw[thick,->] (0,-2) -- (0,2) node[above] {$X^0$};  

\draw[thick, color=gray!90] (0,0) -- (2,2);  
\draw[thick, color=gray!90] (0,0) -- (2,-2);  

\node at (1.7,1.2) {\large$\mathscr{W}$};

\fill[fill=blue!50, opacity=0.3] (1,0) -- (3,2) -- (3,-2) -- cycle;  

\draw[thick, color=blue!90] (1,0) -- (3,2);  
\draw[thick, color=blue!90] (1,0) -- (3,-2);  

\node at (3.2,1.2) {\large$\mathscr{W} + \epsilon$};
\node at (1,0) {$\bullet$};
\node at (0.9,-0.2) {$\epsilon$};
\node at (-0.2,-0.2) {$O$};

\end{tikzpicture}

    \caption{A depiction of the Rindler wedge $\mathscr{W}$ (in grey) translated by $\epsilon$ from the origin $O$ (in blue).}
    \label{fig:wedge}
\end{figure}

Now, as we are considering our spacetime to be $\mathbb{M}_2$, spatial translations by $\epsilon$ are an isometry and thus, by functoriality, have a corresponding algebraic monomorphism $\alpha_\epsilon = \gamma_\epsilon$. This morphism, in turn, induces the dual map $\alpha^*_\epsilon = \gamma^*_\epsilon$ between the state spaces, as equation \eqref{eq:dual_definition}. Hence, by the computation in \eqref{eq:computation_bell_functor}, for a state that (maximally) violates Bell's inequalities in the wedge $\mathscr{W}$, there is a ``corresponding'' state on the algebra located in $\mathscr{W} + \epsilon$ which also (maximally) violates the inequalities.

\section{Speculations and Outlook}\label{sec:speculations}
In this section, we propose some paths that can be taken given the previous results. We mention that this simple result can be further explored by considering more modern formulations of Bell's inequalities in algebraic quantum field theory \cite{Summers:1987ze} using the Murray-von Neumann classification of factors \cite{murray_vonNeumann}, Connes classification of type III factors \cite{connes_factors} and more general spacetime regions.

\subsection{Non-commutative spacetimes}

There is a close relationship between non-commutative $C^*$-algebras and non-commuta\-tive spacetimes in the sense that any such algebra uniquely generates such a spacetime \cite{connes1,connes2,Aschieri:2009zz}.

This equivalence may allow us to translate the algebraic form of Bell's inequalities (equation \eqref{eq:bell_ineq}) into a ``geometric inequality'' in the domain of the LCQFT functor. In other words, one could take the domain category of the LCQFT functor to be composed of non-commutative spacetimes as objects and associate them uniquely with the $C^*$-algebra of the codomain of such functor, meaning we define our domain category in terms of the codomain. A corresponding inequality for the spacetime might be an indication that, at least for this class of manifolds, Bell's inequalities are fundamentally related to the structure of such spacetimes.

This might be a step towards the research programs which treat spacetime as emergent from quantum correlations \cite{Seiberg:2006wf,PhysRevD.95.024031} or towards the program known as the ``It from Qubit'' program \cite{Maldacena:2013xja,VanRaamsdonk:2010pw}. This idea has already been explored in \cite{Doplicher1995} where the connection with von Neumann algebras is made more explicit in the form of ``spacetime uncertainty relations''.

The main challenge with this idea is to prove that the category of non-commutative spacetimes (say, $\mathfrak{NonCommSp}$) is a suitable domain category for the LCQFT functor, in the sense proposed by Benini et. al. \cite{Benini:2017fnn} and generalized by Grant-Stuart \cite{Grant-Stuart:2022xdh}. In particular, since the concept of global hyperbolicity is senseless in the context of non-commutative spacetimes, there is still a challenge to define the morphisms in this hypothetical category after choosing which structures should be preserved and are relevant in the context of quantum field theory.

\subsection{Conformal invariance}
Conformal invariance has played a significant role in theoretical physics \cite{DiFrancesco:1997nk}. In particular, its extensive use in quantum field theory has allowed insights into quantum gravity and strongly coupled systems \cite{Maldacena:1997re,Witten:1998qj}.

In two dimensions, the wedge regions and causal diamond regions given by $X^1 > 0$, $X^+ \geq 0$, $X^-\geq 0$ (with $X^{\pm} = X^1 \pm X^0$) and $|x| + |t| \leq R$ respectively (see figure \ref{fig:conformal}) are related by a conformal transformation (\cite{haag}, section V.4.2), namely
\begin{equation}
    x^\mu = 2R\frac{X^\mu - b^\mu X^2}{1 - 2b\cdot X + b^2 X^2} + Rb^\mu,
\end{equation}
where $b^\mu = (0,-1)$.

\begin{figure}[h!]
    \begin{tikzpicture}

\fill[fill=yellow!30] (0,0) -- (2,2) -- (2,-2) -- cycle;  

\draw[thick,->] (-0.5,0) -- (2,0) node[right] {$X^1$};  
\draw[thick,->] (0,-2) -- (0,2) node[above] {$X^0$};  

\draw[thick] (0,0) -- (2,2);  
\draw[thick] (0,0) -- (2,-2);  

\node at (1.3,0.4) {\huge$\mathcal{R}$};  
\node at (1.3,2.2) {$X^0 = X^1$};  

\fill[fill=yellow!30] (4.5,0) -- (6,1.5) -- (7.5,0) -- (6,-1.5) -- cycle;  

\draw[thick,->] (4,0) -- (8,0) node[right] {$X^1$};  
\draw[thick,->] (6,-2) -- (6,2) node[above] {$X^0$};  

\draw[thick] (4.5,0) -- (6,1.5) -- (7.5,0) -- (6,-1.5) -- cycle;  

\node at (6.5,0.4) {\huge$\mathcal{D}$};  

\draw[<->, thick] (2.8,0) -- (3.6,0);  

\end{tikzpicture}
    \caption{The relation between the (Rindler) wedge region and a causal diamond region is given by a conformal transformation.}
    \label{fig:conformal}
\end{figure}
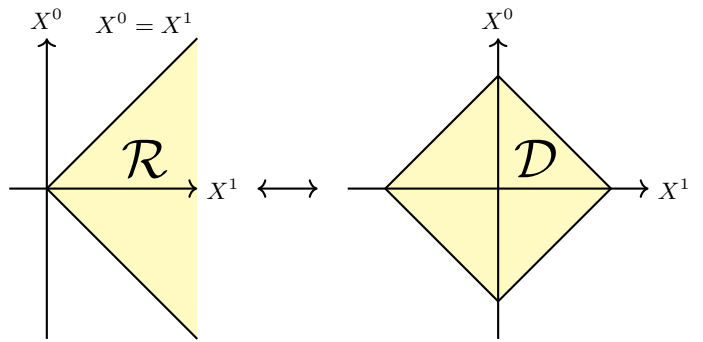
The generalization of the work in \cite{summers&werner1} and \cite{summers&werner2} to causal diamond regions in \cite{summers&werner_tangent}, together with the functorial formulation presented in this paper suggest that it might be possible to extend the morphisms in the domain category of the LCQFT functor (i.e., the $\mathfrak{GlobHyp}$ category of globally hyperbolic spacetimes) to include conformal transformations.

This generalization would be useful in proving the consistency of the formulation presented here by showing that it leads to the same result obtained in \cite{Summers:1987ze}. It would perhaps also be useful for describing holography and phenomena in conformally related regions in terms of locally covariant quantum field theory.

\subsection{Cosmology}

It is speculated in some inflationary models of the early universe that there was an initial singularity \cite{Hawking_Ellis_1973}. This would lead to a topological defect in the beginning of spacetime, making the globally hyperbolic assumption of LCQFT invalid for describing our world. Nevertheless, violations of the inequalities have been successfully verified \cite{Clauser:1969ny, Clauser:1974tg, Clauser:1978ng}.

If we would expect that LCQFT is only feasible in globally hyperbolic spacetimes, then this would entail that there could be no topological defects in the universe. Furthermore, the experimental verification of the violation of Bell's inequalities together with the functorial formulation presented in this paper which assumes global hyperbolicity would be an indication against the hypothesis of an initial singularity.

\section{Conclusions}
We have introduced a simple but effective formulation of Bell's inequalities and their violation in general spacetimes by means of category theory and the functorial formulation of quantum field theory. In this formulation, it is clear that the inequalities are always present in a theory so long as they can be formulated in terms of an algebra of observables. This is an indication of the fundamental role played by entanglement in the mathematical structure of quantum physics.

Whether or not any quantum field theory (described by a Lagrangian or not) can be formulated non-perturbatively in terms of a suitable algebra (being a $C^*$-algebra, von Neumann algebra or other types) is still an open research question. The generality of the functorial formulation due to Brunnetti, Fredenhagen and Verch \cite{BFV03} might allow for more general objects such as loops or strings to be described by this framework, thus allowing Bell's inequalities to be a probe for such theories.

Furthermore, we proposed some research topics that can be explored using the ideas presented in this paper. New proofs for the validity of the inequalities and explicit calculations for their violations using commutative diagrams is also a route to better understand the ideas presented here.

\begin{acknowledgments}
RG is supported by CAPES, grant number 88882.461730/2019-01.
\end{acknowledgments}

\appendix

\nocite{*}

\bibliography{biblio}

\end{document}